\documentclass[a4paper,USenglish,cleveref, autoref, thm-restate]{lipics-v2021}
\usepackage[utf8]{inputenc}
\usepackage{amsthm}
\usepackage{amssymb}
\usepackage{amsmath}
\usepackage{mathtools}
\usepackage{epsfig}
\usepackage{colonequals}
\usepackage{enumerate}
\usepackage{enumitem}
\usepackage{hyperref}
\usepackage{makecell}

\pdfoutput=1 
\hideLIPIcs  
\nolinenumbers 

\newcommand{\GG}{{\mathcal G}}
\newcommand{\trunc}{\mathrm{trunc}}

\title{An efficient implementation and a strengthening of Alon-Tarsi list coloring method}
\author{Zdeněk Dvořák}{Computer Science Institute (CSI) of Charles University, Prague, Czech Republic}{rakdver@iuuk.mff.cuni.cz}{https://orcid.org/0000-0002-8308-9746}{Supported by project 22-17398S (Flows and cycles in graphs on surfaces) of Czech Science Foundation.}
\authorrunning{Z. Dvořák}
\Copyright{Zdeněk Dvořák}

\EventEditors{John Q. Open and Joan R. Access}
\EventNoEds{2}
\EventLongTitle{42nd Conference on Very Important Topics (CVIT 2016)}
\EventShortTitle{CVIT 2016}
\EventAcronym{CVIT}
\EventYear{2016}
\EventDate{December 24--27, 2016}
\EventLocation{Little Whinging, United Kingdom}
\EventLogo{}
\SeriesVolume{42}
\ArticleNo{23}

\begin{document}
\maketitle

\begin{abstract}
As one of the first applications of the polynomial method in combinatorics, Alon and Tarsi gave
a way to prove that a graph is choosable (colorable from any lists of prescribed size).  We describe an
efficient way to implement this approach, making it feasible to test choosability of graphs
with around 70 edges.  We also show that in case that Alon-Tarsi method fails to show that
the graph is choosable, further coefficients of the graph polynomial provide constraints
on the list assignments from which the graph cannot be colored.  This often enables us to
confirm colorability from a given list assignment, or to decide choosability by testing just a few list assignments.
The implementation can be found at \href{https://gitlab.mff.cuni.cz/dvorz9am/alon-tarsi-method}{https://gitlab.mff.cuni.cz/dvorz9am/alon-tarsi-method}.
\end{abstract}

\section{Introduction}

\emph{List coloring} (first developed by Vizing~\cite{vizing1976} and by Erd\H{o}s et al~\cite{erdosrubintaylor1979})
is a generalization of the ordinary proper graph coloring that naturally arises
in inductive proofs of colorability, but also has a rich and interesting theory of its own.
For a graph $G$, a \emph{list assignment} is a system $L=\{L_v:v\in V(G)\}$ of finite sets
assigned to vertices of $G$, specifying the colors that can be used to color each of the vertices.
An \emph{$L$-coloring} is a function $\varphi\in \bigtimes_{v\in V(G)} L_v$ such that $\varphi(u)\neq\varphi(v)$ for every edge $uv\in E(G)$.
In particular, if $L$ assigns to all vertices the same list of $k$ colors, then an $L$-coloring is exactly a proper coloring using $k$ colors.
We say that a graph $G$ is \emph{$k$-choosable} if $G$ has an $L$-coloring for every list assignment $L$ such that $|L_v|=k$ for
every $v\in V(G)$.  Clearly, a $k$-choosable graph is $k$-colorable, but the converse is false---for every $k$, there exists
a bipartite graph that is not $k$-choosable.  Moreover, while all planar graphs are 4-colorable~\cite{AppHak1,AppHakKoc}, there
exist non-4-choosable planar graphs~\cite{voigt1993} (but all planar graphs are 5-choosable~\cite{thomassen1994}).

As we mentioned, list coloring naturally appears in inductive proofs of colorability by the method
of reducible configurations.  In this setting, it is necessary to allow the list of each vertex to be of different size.
For a function $s:V(G)\to\mathbb{N}$, an \emph{$s$-list-assignment} is a list assignment $L$ such that $|L_v|\ge s(v)$ for
each vertex $v\in V(G)$, and we say that $G$ is \emph{$s$-choosable} if it is $L$-colorable for every $s$-list-assignment $L$.
The method of reducible configurations can be summarized as follows.  Suppose $\GG$ is a subgraph-closed class of graphs
and we want to show that all graphs in $\GG$ are $k$-colorable.  Suppose $C$ is an induced subgraph of a graph $G\in \GG$
and for $v\in C$, let $s'(v)=k-(\deg_G v - \deg_C v)$.  If $C$ is $s'$-choosable, then to show that $G$ is $k$-colorable,
it suffices to show that $G-V(C)$ is $k$-colorable.  Indeed, given a $k$-coloring $\varphi$ of $G-V(C)$, for each $v\in V(C)$ let $L_v$
be the set of colors that are not used on the neighborhood of $v$. Clearly $|L_v|\ge s(v)$ for each $v\in V(C)$, and thus
$C$ is $L$-colorable. Moreover, the choice of $L$ ensures that any $L$-coloring of $C$ combines with $\varphi$ to a proper $k$-coloring of the whole graph $G$.
Hence, if we show that every graph in $\GG$ contains an induced subgraph with this property, this shows that all graphs in $\GG$
are $k$-colorable, and actually, even $k$-choosable.  There is a huge number of papers using this approach or its variations; see~\cite{cranston2017introduction}
for a general introduction.

To execute this kind of arguments, one needs to be able to verify that (usually relatively small) graphs $C$ are $s$-choosable for specified
functions $s:V(C)\to\mathbb{N}$.
As long as one uses only a few of these \emph{reducible configurations}, one can show this ad-hoc.  However, in
more complicated proofs the number of reducible configurations may be in hundreds or thousands, and checking their choosability
by hand is not an option.  Hence, it is necessary to come up with a way to check choosability automatically.
Checking whether a given graph $C$ is $L$-colorable for a given list assignment $L$ is of course NP-hard, and no algorithm
with subexponential time complexity is known (or believed to exist); but this is actually not much of an issue:
For typical configurations with 10--20 vertices, standard exponential-time coloring algorithms (or even generic CSP or SAT solvers)
suffice to find an $L$-coloring in a reasonable time.  However, to verify $s$-choosability, one needs to go over all possible
$s$-list-assignments, and their number grows very fast\footnote{Let us remark that it may be feasible to go over all relevant $s$-list-assignments
when the configurations are very small and we only prove reducibility for $k$-coloring, i.e.,  we only need to investigate the list assignments consisting of subsets
of $\{1,\ldots,k\}$ and we can take advantage of the fact that the colors are interchangeable to reduce the number of possible choices.}.
For example, just the number of ways how to assign lists of size one to $n$ vertices is equal to the Bell number $B_n\ge \exp(\Omega(n\log n))$,
which exceeds $10^9$ for $n=15$.
Hence, rather than deciding choosability exactly, one usually employs conservative heuristics that either correctly decide that
a graph $C$ is $s$-choosable, or fail and do not provide any information about choosability of $C$.

A popular heuristic is based on the polynomial method and was developed by Alon and Tarsi~\cite{alon1992colorings}.
For a graph $G$, the \emph{graph polynomial} $P_G$ in variables $x_v$ for $v\in V(G)$
is defined by choosing an orientation $\vec{G}$ of $G$ arbitrarily and letting
$$P_G=\prod_{(u,v)\in E(\vec{G})} (x_v-x_u).$$
Note that the graph polynomial is defined uniquely up to the sign.  Throughout the paper,
by a statement such as ``the graph polynomial of $G$ is $p$'', we mean ``there exists
an orientation of $G$ such that the above definition results in $p$ (and for any other
orientation, it results either in $p$ or in $-p$)''.   For an index set $I$ and a function $f:I\to\mathbb{N}$,
let us define $x^f=\prod_{i\in I} x_i^{f(i)}$.
For a polynomial $q$ in variables $x_i$ for $i\in I$, let $[x^f]\,q$ denote the coefficient of the monomial $x^f$ in $q$.
The \emph{degree} $\deg_{x_i} q$ of a variable $x_i$ in $q$ is $\max\{f(i): f\in\mathbb{N}^I, [x^f]\,q\neq 0\}$.
For $\varphi:I\to\mathbb{R}$, let $q(\varphi)$ denote the evaluation of $q$ at $\varphi$, that is,
$$q(\varphi)=\sum_{f\in\mathbb{N}^I} [x^f]\,q \cdot \prod_{i\in I} \varphi(i)^{f(i)}.$$
The graph polynomial is connected to list coloring through the following simple observation.
\begin{observation}\label{obs-color}
Let $L=\{L_v\subset \mathbb{R}:v\in V(G)\}$ be an assignment of lists to vertices of $G$.
The graph $G$ is $L$-colorable if and only if there exists $\varphi\in \bigtimes_{v\in V(G)} L_v$
such that $P_G(\varphi)\neq 0$.
\end{observation}

Alon-Tarsi choosability method is based on the following algebraic statement known as the Combinatorial Nullstellensatz.
\begin{theorem}\label{thm-nullstellensatz}
Let $q\neq 0$ be a polynomial in variables $x_i$ for $i\in I$, and let $\{S_i\subset \mathbb{R}:i\in I\}$
be a system of finite sets.  If $\deg_{x_i} q < |S_i|$ for every $i\in I$, then there exists $s\in\bigtimes_{i\in I} S_i$
such that $q(s)\neq 0$.
\end{theorem}

With a little extra work, these combine to the following claim.
\begin{theorem}\label{thm-alontarsi}
Let $G$ be a graph and let $s,f:V(G)\to \mathbb{N}^{V(G)}$ be functions such that
$f(v)<s(v)$ for every $v\in V(G)$.
If $[x^f]\,P_G\neq 0$, then $G$ is $s$-choosable.
\end{theorem}

Let us give a simple example of how Theorem~\ref{thm-alontarsi} can be used.

\begin{example}
For $n$ even, the graph polynomial of the $n$-cycle $C_n$ with vertices $0$, \ldots, $n-1$ is $2x_0\cdots x_{n-1}+\ldots$;
hence, setting $f(v)=1$ for every $v\in V(C_n)$, Theorem~\ref{thm-alontarsi} shows that $C_n$ is $2$-choosable.
\end{example}

Of course, to use Theorem~\ref{thm-alontarsi} to show choosability, we need to be able to compute
relevant coefficients of the graph polynomial. A direct computation would require one to go over $2^{|E(G)|}$ possibilities
and as such is not practical for graphs with more than about 30 edges.  As our first contribution, in Section~\ref{sec-impl},
we describe an algorithm to find the relevant coefficients that is more efficient in practice, evaluate
an implementation of this algorithm, and discuss the various choices leading to this algorithm as well as possible improvements.

Theorem~\ref{thm-alontarsi} is a conservative heuristic, in the sense that
if all relevant coefficients of $P_G$ are zero, it does not give us any information about whether $G$ is $s$-choosable or not.
As our second contribution, we employ further coefficients of the graph polynomial to obtain information
about list assignments from which $G$ cannot be colored.  For a list assignment $\{L_v:V\in V(G)\}$,
the \emph{characteristic vector} of a color $c$ is the function $\chi_{c,L}:V(G)\to\{0,1\}$ such that
$$\chi_{c,L}(v)=\begin{cases}
1&\text{if $c\in L_v$}\\
0&\text{if $c\not\in L_v$}
\end{cases}$$
for each $v\in V(G)$.  For $v\in V(G)$, let $1_v:V(G)\to\{0,1\}$ be defined by setting $1_v(v)=1$ and $1_v(u)=0$
for every $u\in V(G)\setminus\{u\}$.

\begin{theorem}\label{thm-main}
Let $G$ be a graph, let $L=\{L_v\subset \mathbb{R}:v\in V(G)\}$ be an assignment of lists to vertices of $G$,
and suppose that $f:V(G)\to \mathbb{Z}$ satisfies $f(v)<|L_v|$ for every $v\in V(G)$ and $\sum_{v\in V(G)} f(v)=|E(G)|-1$.
If $G$ is not $L$-colorable, then for every color $c$,
$$\sum_{z\in V(G)} [x^{f+1_z}]\,P_G\cdot \chi_{L,c}(z) = 0.$$
\end{theorem}
Theorem~\ref{thm-main} is proved by a straightforward modification of the proof of Theorem~\ref{thm-alontarsi} by Alon and Tarsi~\cite{alon1992colorings},
and we postpone it till Section~\ref{sec-gener}, where we also discuss implementation details and testing results.
Here, let us give several examples of the conclusions that can be drawn from Theorem~\ref{thm-main}.

\begin{example}\label{ex-cycle}
The graph polynomial of the 5-cycle $C_5$ with vertices $0$, \ldots, $4$ in order is 
\begin{align*}
 &x_0x_1x_2x_3(x_0-x_1+x_2-x_3)\\
+&x_0x_1x_2x_4(-x_0+x_1-x_2+x_4)\\
+&x_0x_1x_3x_4(x_0-x_1+x_3-x_4)\\
+&x_0x_2x_3x_4(-x_0+x_2-x_3+x_4)\\
+&x_1x_2x_3x_4(x_1-x_2+x_3-x_4)\\
+&\ldots
\end{align*}
Applying Theorem~\ref{thm-main} for $f$ being $(1,1,1,1,0)$, $(1,1,1,0,1)$, $(1,1,0,1,1)$, $(1,0,1,1,1)$, and $(0,1,1,1,1)$,
we conclude that the characteristic vector $\chi$ of each color satisfies the following system
of linear equations:
\begin{align*}
\chi(0)-\chi(1)+\chi(2)-\chi(3)&=0\\
-\chi(0)+\chi(1)-\chi(2)+\chi(4)&=0\\
\chi(0)-\chi(1)+\chi(3)-\chi(4)&=0\\
-\chi(0)+\chi(2)-\chi(3)+\chi(4)&=0\\
\chi(1)-\chi(2)+\chi(3)-\chi(4)&=0
\end{align*}
This system has only two solutions in $\{0,1\}^{V(C_5)}$, namely $\chi=(0,0,0,0,0)$ or $\chi=(1,1,1,1,1)$.
That is, if $C_5$ is not $L$-colorable and $|L_v|=2$ for each $v\in V(C_5)$, then every color that is used
in $L$ must appear in all the lists, and thus $L_0=L_1=L_2=L_3=L_4$ (and, we can check that $G$
indeed is not $L$-colorable from such a list assignment).
\end{example}

\begin{example}\label{ex-fan}
Let $G$ consist of the path $1234$ and a vertex $0$ adjacent to all the vertices of this path,
and consider a list assignment $L$ such that $|L_0|=4$ and $|L_1|=|L_2|=|L_3|=|L_4|=2$.
The graph polynomial of $G$ is
\begin{align*}
&x_0^2x_1x_2x_3x_4(x_1-x_2+x_3-x_4)\\
+&x_0^3x_2x_3x_4(-x_0+x_2+x_4)\\
+&x_0^3x_1x_3x_4(x_0-x_1-x_4)\\
+&x_0^3x_1x_2x_4(-x_0+x_1+x_4)\\
+&x_0^3x_1x_2x_3(x_0-x_1-x_3)\\
+&\ldots
\end{align*}
Hence, Theorem~\ref{thm-main} gives the following system of linear equations for the characteristic vectors $\chi$.
\begin{align*}
\chi(1)-\chi(2)+\chi(3)-\chi(4)&=0\\
-\chi(0)+\chi(2)+\chi(4)&=0\\
\chi(0)-\chi(1)-\chi(4)&=0\\
-\chi(0)+\chi(1)+\chi(4)&=0\\
\chi(0)-\chi(1)-\chi(3)&=0.
\end{align*}
This implies $\chi(1)=\chi(2)$, $\chi(4)=\chi(3)$, and $\chi(0)=\chi(2)+\chi(3)$.  Since $\chi(0)\le 1$, this
means that $\chi(2)$ and $\chi(3)$ cannot both be one, and thus if $G$ is not $L$-colorable, then $L_2\cap L_3=\emptyset$.
Therefore, $G$ is $L$-colorable iff the graph $G'$ obtained from $G$ by deleting the edge $23$ is $L$-colorable.

The non-zero solutions to the system of equations in $\{0,1\}^5$ are $(1,1,1,0,0)$ and $(1,0,0,1,1)$. 
There is only one way how to combine these vectors to obtain the prescribed list sizes $(4,2,2,2,2)$,
implying that if $G$ is not $L$-colorable, then (up to permutation of colors) $L_1=L_2=\{a,b\}$, $L_3=L_4=\{c,d\}$,
and $L_0=\{a,b,c,d\}$.  We can easily check that indeed $G$ is not $L$-colorable.
\end{example}

\begin{example}\label{ex-wheel}
Let $G$ be the wheel with center $0$ and spokes $1$, \ldots, $5$, and let $s(0)=5$, $s(1)=s(2)=s(3)=2$ and $s(4)=s(5)=3$.
All monomials of $P_G$ with degrees smaller than those given by $s$ are zero, and thus Theorem~\ref{thm-alontarsi}
does not give any conclusion about $s$-choosability of $G$.
The only non-zero characteristic vectors satisfying the system of linear equations determined using Theorem~\ref{thm-main}
are $(1,1,1,1,0,0)$ and $(1,0,0,0,0,1,1)$.  This gives only one possibility for the $s$-list-assignment $L$ such that $G$ may
not be $L$-colorable (up to permutation of colors), namely
$L_1=L_2=L_3=\{a,b\}$, $L_4=L_5=\{c,d,e\}$, and $L_0=\{a,b,c,d,e\}$.  However, we can easily check that $G$ is $L$-colorable,
and thus $G$ is $s$-choosable.

Alternately, note that the characteristic vectors of the colors imply that any $s$-list-assignment $L$
such that $G$ is not $L$-colorable satisfies $L_3\cap L_4=\emptyset$ and $L_5\cap L_1=\emptyset$.
Consequently, $G$ is $L$-colorable iff the graph $G'$ obtained by deleting the edges $34$ and $51$ is $L$-colorable.
However, Theorem~\ref{thm-alontarsi} shows that $G'$ is $s$-choosable, and thus $G$ is $s$-choosable as well.
\end{example}

\begin{example}\label{ex-nolist}
Let $G$ consist of the wheel with center $0$ and spokes $1$, \ldots, $5$ and additional vertices $6$ adjacent to $1$ and $2$
and $7$ adjacent to $2$ and $3$, with a list assignment $L$ such that $|L_1|=|L_2|=|L_3|=3$, $|L_4|=|L_5|=|L_6|=|L_7|=2$
and $|L_0|=5$.  If $G$ is not $L$-colorable, then according to Theorem~\ref{thm-main}, the feasible characteristic vectors of colors are
$(0,0,0,0,0,0,0)$ and $(1,1,1,1,1,0,0)$.  There is clearly no way how to compose $L$ from colors with these characteristic vectors,
and thus $G$ is $L$-colorable.
\end{example}

Thus, there are a number of ways the outcome of Theorem~\ref{thm-main} can give interesting information about $L$-colorability
of a graph $G$, in the case that Theorem~\ref{thm-alontarsi} fails (by \emph{feasible characteristic vectors}, we mean
the solutions to the system of linear equations obtained using Theorem~\ref{thm-main} that belong to $\{0,1\}^{V(G)}$):
\begin{itemize}
\item It may happen that no linear combination of the feasible characteristic vectors with non-negative integral
coefficients is equal to the vector $s$ of the sizes of the lists, showing that $G$ is $s$-choosable (Example~\ref{ex-nolist}).
\item It may happen that the feasible characteristic vectors can be combined to give lists of the prescribed size only
in a small number of ways, and by going over the corresponding $s$-list-assignments, we can
\begin{itemize}
\item prove that $G$ is $s$-choosable (Example~\ref{ex-wheel}), or
\item identify all $s$-list-assignments $L$ such that $G$ is not $L$-colorable (Examples~\ref{ex-cycle} and~\ref{ex-fan}).
\end{itemize}
\item It may be possible to use the characteristic vectors to conclude that if $G$ is not $L$-colorable from an $s$-list-assignment $L$, then the lists assigned to adjacent
vertices $u$ and $v$ must be disjoint (Examples~\ref{ex-wheel} and \ref{ex-fan}).  Consequently, $G$ is $L$-colorable iff $G-uv$ is $L$-colorable.
This can be all that we need (e.g., this is sufficient to execute the argument outlined in the description of the method of reducible configurations above,
just instead of deleting the whole reducible configuration, we only delete this edge $uv$).
Alternately, we can iteratively process $G-e$ and conclude that it is $s$-choosable (Example~\ref{ex-fan}), and thus $G$ is $s$-choosable
as well.
\end{itemize}

Unlike Theorem~\ref{thm-alontarsi}, using Theorem~\ref{thm-main} we can conclude that a graph $G$ is $L$-colorable from
a specific $s$-list-assignment $L$, even if $G$ is not $s$-choosable.
It might be possible to use this to speed up backtracking algorithms to find an $L$-coloring:
Whenever we branch, we use the conditions from Theorem~\ref{thm-main} to check whether in one of the branches,
the rest of the graph is guaranteed to be colorable from the lists of colors still available for them.
If that is the case, we proceed just to this branch and cut off the rest.  It should be noted however that the overhead of
this extra test may outweigh the benefit of reducing the search tree.

Of course, there are (even quite simple) graphs for which Theorem~\ref{thm-main} fails to provide any information.
The method can be pushed further by considering more coefficients of the graph polynomial; we discuss this in more detail
in Section~\ref{sec-more}.

\section{Efficient implementation of Alon-Tarsi method}\label{sec-impl}

In theoretical applications of Theorem~\ref{thm-alontarsi}, one usually guesses which monomial(s) of the graph polynomial are
relevant, then shows that at least of them appears with a non-zero coefficient, possibly using the following standard combinatorial interpretation
of the coefficient shown already by by Alon and Tarsi~\cite{alon1992colorings}.
For a function $f:V(G)\to\mathbb{N}$, an \emph{$f$-orientation} of $G$ is an orientation of the edges of $G$ such that
each vertex $v\in V(G)$ has outdegree exactly $f(v)$.  Given a fixed reference orientation of $G$,
the \emph{sign} of another orientation $G$ that differs from the reference orientation in $m$ edges is $(-1)^m$.

\begin{lemma}\label{lemma-orient}
Let $G$ be a graph.  For any $f:V(G)\to\mathbb{N}$, the coefficient at $x^f$ in the graph polynomial of $G$
is the sum of the signs of the $f$-orientations of $G$.
\end{lemma}

Sum of signs of orientations superficially resembles the definition of a determinant, and thus one could
hope that it might be possible to compute it in polynomial time.  However, this does not seem to be the
case.  Note that for bipartite graphs, all $f$-orientations have the same sign (as they differ only by
flipping edges of an Eulerian subgraph).  Consider a 4-regular graph $G$, and let $G'$ be the graph
obtained from $G$ by subdividing each edge once.  Let $f$ assign $2$ to each vertex of $V(G)$ and $1$ to each
vertex of $V(G')\setminus V(G)$.  Then $|[x^f]\,P_{G'}|$ is the number of $f$-orientations of $G'$, which
is equal to the number of Eulerian orientations of $G$.  However, computing the number of Eulerian orientations
of a 4-regular graph is \#P-hard~\cite{huang2012dichotomy}, even for planar graphs~\cite{guo2013complexity}.
Hence, evaluating the coefficients of the graph polynomial is \#P-hard for bipartite planar graphs of maximum degree four.

To apply Theorem~\ref{thm-alontarsi}, one actually only needs to decide whether one of the relevant coefficients
is non-zero, not to evaluate them exactly, and we have not excluded the possibility that this could be done in
polynomial time.  However, this appears to be unlikely in general graphs\footnote{Though it should be mentioned
that in bipartite graphs, since all $f$-orientations have the same sign, this becomes just the question of whether
there exists an $f$-orientation for a function $f$ satisfying $f(v)<s(v)$ for each vertex $v$, which
can be answered in polynomial time through a standard reduction to the maximum matching in bipartite graphs.}.
Moreover, to apply Theorem~\ref{thm-main}, we actually need to determine the values of the coefficients.

Hence, let us turn out our attention to the problem of designing an exponential-time, yet practically useful algorithm
to compute the relevant coefficients of the graph polynomial.  We start by discussing some options, gradually arriving at
our chosen solution.

\subsection{Direct enumeration}\label{sec-direct}

The first approach that one might consider to determine a particular coefficient $[x^f]\,P_G$ is by a direct
enumeration of all $f$-orientations of $G$.  While this approach turns out not to be efficient enough, let us discuss
it in more depth before proceeding to a more useful algorithm.

It is possible to enumerate all $f$-orientations of a graph $G$ in polynomial time per each $f$-orientation.
To do so, it suffices to observe that it is possible to determine in polynomial time whether a partial
orientation extends to an $f$-orientation, by a straightforward reduction to maximum matching
in an auxiliary bipartite graph (obtained by subdividing the edges that have not yet been oriented, and
blowing up each original vertex $v$ into $f(v)$ minus the current outdegree of $v$ vertices).
With this subroutine, we can try orienting each of the edges of $G$ in both ways by backtracking,
immediately cutting off branches that do not lead to at least one $f$-orientation.

In practice, branches are cut off fairly rarely, and thus it is not worth the overhead of running the extendability test.
It turns out to be more efficient to just cut off trivially useless branches (where the
outdegree of a vertex $v$ exceeds $f(v)$), in combination with a suitable rule for selecting the edge
to orient (e.g. choosing an edge incident with a vertex $v$ whose current outdegree is closest to $f(v)$,
and in particular preferring the edges whose orientation is forced, as they are incident with a vertex $v$ with outdegree
equal to $f(v)$).

An issue with this approach is that in general the number of $f$-orientations to list grows quite fast with the number of edges.
Moreover, one may be forced to go over all possible choices of $f$ dominated by the list sizes.
\begin{example}\label{ex-all}
Consider the clique $K_n$ with lists of size $n-1$.  To verify that Theorem~\ref{thm-alontarsi} does not apply in this situation
(which it cannot, since $K_n$ is not $(n-1)$-colorable), we would have to compute all coefficients $[x^f]\,K_n$ for all functions $f$
such that $f(v)\le n-2$ for each $v\in V(K_n)$.  To do so, we would list all orientations of $K_n$ except for those that contain
a vertex of outdegree $n-1$, and the number of such orientations is
$$2^{|E(K_n)|}-n2^{|E(K_n)|-n+1}=2^{|E(K_n)|}\bigl(1-\frac{n}{2^{n-1}}\bigr).$$
Even with a very efficient implementation, this becomes impractically slow when the number of edges exceeds about $40$.
\end{example}
Example~\ref{ex-all} is basically the worst case for the direct enumeration algorithm.  Still, our experiments show
that even on less artificial inputs, the algorithm becomes quite slow for graphs of this size (taking several minutes
to process a graph with $40$ edges).  Let us mention a slight advantage of computing the coefficients one by one:
We can stop immediately once we find a non-zero one, improving the efficiency of the approach in case the answer is positive.

\subsection{Truncated multiplication}

A more efficient approach follows from an easy observation showing that we can process the edges one by one
if we compute all relevant coefficients at once.
For two functions $f,s:I\to \mathbb{N}$, we write $f\prec s$ if $f(i)<s(i)$ for each $i\in I$.
For a polynomial $q$ in variables $\{x_i:i\in I\}$ and a function $s:I\to \mathbb{N}$, let
$$\trunc_s(q)=\sum_{s'\prec s} [x^{s'}]\,q\cdot x^{s'},$$
i.e., the polynomial $\trunc_s(q)$ is obtained from $q$ by truncating it to monomials where the degree of each variable $x_i$
is less than $s(i)$.

\begin{observation}\label{obs-grad}
For any polynomials $q$ and $m$ in variables $\{x_i:i\in I\}$ and a function $s:I\to \mathbb{N}$,
$$\trunc_s(q\cdot m)=\trunc_s(\trunc_s(q)\cdot m).$$
\end{observation}

In particular, Theorem~\ref{thm-alontarsi} states that if $\trunc_s(P_G)\neq 0$, then $G$ is $s$-choosable,
and Observation~\ref{obs-grad} shows that we can compute $\trunc_s(P_G)$ by multiplying the terms $(x_v-x_u)$ for $uv\in E(G)$
one by one and performing the truncation after each multiplication.  In the implementation, we process the edge
$uv$ when the current truncated polynomial is $q$ is as follows: For each monomial $a\cdot x^f$ of $q$ with $a\neq 0$,
\begin{itemize}
\item[(i)] if $f(v) + 1 < s(v)$, then add $a\cdot x^{f+1_v}$ to the output, and
\item[(ii)] if $f(u) + 1 < s(u)$, then add $-a\cdot x^{f+1_u}$ to the output.
\end{itemize}
Of course, we represent the polynomial $q$ as a table mapping each vector $f$ of the degrees to the coefficient $[x^f]\,q$,
so this amounts just to traversing the table representing $q$ and adding the values to the table representing the output.
Note that it may (and often does) happen that the coefficients cancel out---if $[x^{f-1_u}]\,q=[x^{f-1_v}]\,q$, then the
contributions to the coefficient at $x^f$ in the output sum to zero, and the monomial needs to be deleted from the output.

\begin{example}
As a quick qualitative comparison with the direct enumeration, let us consider the example of the clique $K_n$ with
list sizes $n-1$.  The number of functions $f:V(K_n)\to\mathbb{N}$ such that $f(v)\le n-2$ for each $v\in V(K_n)$
is at most $(n-1)^n=2^{n\log_2 n}$.  Assuming we have memory available on the order of gigabytes, we will be able
to store the associated coefficients as long as $n\log_2 n$ is at most about 30, i.e., for $n$ up to about $10$,
at which point $K_n$ has $45$ edges.
\end{example}
This is of course a crude overestimate, not taking into account the fact that we do not store all coefficients at once
(after $b$ edges were processed, the degree function $f$ of each monomial satisfies $\sum_{v\in V(G)} f(v)=b$),
and that some of the monomials turn out to have coefficient zero and can be dropped.  In our experiments,
a straightforward implementation of the truncated multiplication method starts to run out of memory for graphs with about 50--60 edges.
At this point the program runs for less than a minute, so time is not the limitation here.  We describe a way to decrease the memory
requirements in Section~\ref{ssec-seq}.  Before that, let us give a few remarks.

Similarly to the direct enumeration approach, in addition to deleting the monomials with zero
coefficients, we can also delete the monomials whose multiplication with the rest of the terms of the graph polynomial
cannot result in a monomial with degrees bounded by $s$; monomials with this property can be identified using the
same reduction to the maximum matching in bipartite graphs.  However, our testing shows that 
adding this overhead is not worthwhile except possibly in some very rare circumstances; more on this in Section~\ref{sec-testing}.

Let us also note that unlike the direct enumeration approach, we cannot terminate early, as the non-zero coefficients
are only computed all at once when all the edges have been processed.  Hence, the truncated multiplication
approach might be expected to be less efficient in ``easy to color graphs'' where the graph polynomial has
many applicable non-zero coefficients.  The improvement given in Section~\ref{ssec-seq}, whose primary purpose is to decrease the memory
requirements, also mitigates this issue.

Before we describe this improvement, let us describe an important detail in the efficient implementation of the truncated
multiplication algorithm.

\subsection{Representation of the polynomial}\label{ssec-repr}

To execute the truncated multiplication algorithm, we need to be able to add to coefficients of specific monomials
in the output polynomial.  Hence, it seems natural to represent the polynomial as an associative array mapping
the vector of the degrees of a monomial to the coefficient.  However, each of the standard ways of implementing
associative arrays (hash tables, search trees, tries, \ldots) comes with substantial time and space overheads.
Fortunately, a much simpler and more efficient alternative exists.

Let $<$ be the lexicographic ordering on functions from $V(G)$ to $\mathbb{N}$; i.e., we fix the order $v_1$, \ldots, $v_n$
of vertices of $G$ arbitrarily and we write $f_1<f_2$ if there exists $a\in\{1,\ldots,n\}$ such that $f_1(v_i)=f_2(v_i)$
for $i\in\{1,\ldots,a-1\}$ and $f_1(v_a)<f_2(v_a)$.  We store a polynomial $q$ as the list of all pairs $(f,c)$
such that $[x^f]\,q = c\neq 0$, sorted in the increasing lexicographic ordering according to $f$.
The key observation is as follows: If $f_1<f_2$, then for any vertex $v\in V(G)$, $f_1+1_v<f_2+1_v$.
That is, if we process the monomials in the input polynomial in the increasing lexicographic order of their degree vectors,
then the monomials added to the output polynomial in the part (i) of the truncated multiplication algorithm
are also produced in the increasing lexicographic order, and so are the monomials produced in part (ii).
Hence, we can merge the coefficients produced in (i) and (ii) as in mergesort.

That is, with the representation by a sorted list, the truncated multiplication algorithm can be implemented
as follows:  Let $(f_u,c_u)$ and $(f_v,c_v)$ be the earliest elements in the input list such that
$f_u(u)+1 < s(u)$ and $f_v(v)+1 < s(v)$.  Perform the following operation until both of these elements reach
the end of the input list:
\begin{itemize}
\item If $f_u+1_u=f_v+1_v$, then:
\begin{itemize}
\item If $c_u\neq c_v$, add $(f_v+1_v,c_v-c_u)$ to the end of the output.
\item Advance $(f_u,c_u)$ and $(f_v,c_v)$ to the next input elements such that $f_u(u)+1 < s(u)$ and $f_v(v)+1 < s(v)$.
\end{itemize}
\item Otherwise, if $f_u+1_u>f_v+1_v$, then add $(f_v+1_v,c_v)$ to the end of the output
and advance $(f_v,c_v)$ to the next input element such that $f_v(v)+1 < s(v)$.
\item Otherwise, add $(f_u+1_u,-c_u)$ to the end of the output and advance $(f_u,c_u)$ to the next input element such that $f_u(u)+1 < s(u)$.
\end{itemize}

As the lists only need to be accessed sequentially and we only need to add to the end,
we store them simply as arrays; in addition to simplicity, this results in a cache-friendly
memory access pattern.  To minimize the memory requirements, we store the degree vectors in a packed way,
taking only $\lceil \log_2 \max \{s(v):v\in V(G)\}\rceil\cdot |V(G)|$ bits\footnote{In reality, slightly
more, as we add padding to ensure that the bits representing the degree of each vertex are stored
in the same memory word, rather than being potentially split across two consecutive words in case that
$\lceil \log_2 \max \{s(v):v\in V(G)\}\rceil$ does not divide the width of the word.  This simplifies
operations such as the addition of $1_v$ for some vertex $v\in V(G)$ to the degree vector.}.
In addition to taking less memory, this minimizes the time complexity of copying the (modified)
degree vectors from the input to the output.

For illustration, the change to this mergesort-like approach from our initial implementation using
a hash table (\texttt{unordered\_map} from the standard C++ library) improved the performance by the factor of about $10$.

\subsection{Sequentialization (and parallelization)}\label{ssec-seq}

The main limitation of the truncated multiplication algorithm is the memory consumption---already for graphs with
around 50 edges, it commonly consumes gigabytes of memory.  The key problem is that we compute all the
coefficients of the truncated graph polynomial (and the intermediate polynomials) at once.
To avoid this issue and to divide the work into smaller chunks, we use the following observation.
We say that two polynomials $p$ and $q$ are \emph{disjoint} if no monomial appears with a non-zero
coefficient in both of them.  For a set $S$ of vertices of a graph $G$, we say that two vectors $f,f':V(G)\to \mathbb{N}$
are \emph{$S$-equivalent} if $f(v)=f'(v)$ for every $v\in S$.  We say that a polynomial $q$ in
in variables $\{x_v:v\in V(G)\}$ is \emph{$S$-homogeneous} if any two vectors $f$ and $f'$ such that $[x^f]\,q\neq 0$
and $[x^{f'}]\,q\neq 0$ are $S$-equivalent.  The \emph{$S$-partition} of a polynomial $p$ in variables $\{x_v:v\in V(G)\}$,
is the smallest system $p_1$, \ldots, $p_m$ of pairwise disjoint $S$-homogeneous polynomials such that $p=p_1+\ldots+p_m$.
\begin{observation}
Let $H$ be a spanning subgraph of a graph $G$, let $S\subseteq V(G)$ consist of vertices such that
all incident edges of $G$ belong to $E(H)$, and let $H'$ be the spanning subgraph of $G$ with edge set $E(G)\setminus E(H)$.
Let $s:V(G)\to\mathbb{N}$ be an arbitrary function
and let $p_1$, \ldots, $p_m$ be the $S$-partition of $\trunc_s(P_H)$.  For $i=1,\ldots, m$, let
$q_i=\trunc_s(p_i\cdot P_{H'})$.  Then $q_1$, \ldots, $q_m$ is the $S$-partition of $\trunc_s(P_G)$.
\end{observation}
In particular, the results $\trunc_s(p_1\cdot P_{H'})$, \ldots, $\trunc_s(p_m\cdot P_{H'})$
are pairwise disjoint, and thus these truncated multiplications can be performed completely independently.

This suggests the following algorithm.  Fix an ordering $v_1$, $v_2$, \ldots, $v_n$ of the vertices of $G$.
Start with the polynomial $p=1$, and for $i=1$, $2$, \ldots in order, apply the truncated multiplication algorithm for the not-yet-processed
edges incident with $v_i$.  After each vertex $v_i$ is processed, if the number of monomials of $p$ with
non-zero coefficients exceeds some bound $N$, stop this process,
\begin{itemize}
\item divide $p$ into its $\{v_1,\ldots,v_i\}$-partition $p_1$, \ldots, $p_m$, and
\item for $j=1,\ldots, m$, run the same procedure recursively starting from $p_j$ and processing the
edges of the subgraph induced by $\{v_{i+1}, \ldots, v_n\}$.
\end{itemize}
Let us remark that in the representation of $p$ described in the previous section,
the $\{v_1,\ldots,v_i\}$-equivalent monomials appear consecutively in the list, and thus we
can easily break up $p$ into its $\{v_1,\ldots,v_i\}$-partition.
For the standard Alon-Tarsi method, if we end up with a non-zero polynomial in any of the branches, we can stop
the whole process immediately, recovering the one advantage of the direct enumeration algorithm.
For the extended Alon-Tarsi method (Theorem~\ref{thm-main}), we need to collect the information from all the branches;
we discuss this in more detail in Section~\ref{sec-gener}.

The effect of the choice of the bound $N$ can be seen from the results of the following evaluation
on a non-choosable graph with 25 vertices and 60 edges
(six copies of $K_5$ glued in a path-like fashion over distinct vertices, with the list
sizes equal to the vertex degrees).

\begin{center}
\begin{tabular}{|l|c|c|c|c|c|}
\hline
$N$&$10^3$&$10^4$&$10^5$&$10^6$&$10^7$\\
\hline
time&30 s&25 s&23 s&24.5 s&25.5 s\\
\hline
\end{tabular}
\end{center}

The optimal value of $N$ at the testing machine seems to be around $10^5$: With smaller $N$, we need to pay the overhead
of the branching routine more often, and we do not take full advantage of the efficient sequential
mergesort-like processing.  With larger values of $N$, we no longer fit in the L2 cache.

Let us remark that since the computations in the branches are independent, they can be performed in
parallel or even in a distributed fashion; we did not implement these improvements (one reason
is that we are mostly interested in the extended Alon-Tarsi method, and in that setting the task
of combining the information obtained from different branches in non-sequential fashion seems
somewhat non-trivial).

\subsection{Choice of the ordering}\label{ssec-order}

In what order should the edges of the graph processed?  It is clearly beneficial if the truncation
eliminates as many coefficients as early as possible, and thus we should start from edges in
hardest to orient subgraphs, that is, in subgraphs $H$ with $\Bigl(\sum_{v\in V(H)} s(v)\Bigr) - |E(H)|$
minimum.  So, ideally we would like to fix an ordering $v_1$, \ldots, $v_m$ of vertices
chosen so that such hard-to-orient subgraphs appear early in the ordering,
and then for $a=2, \ldots, m$ process edges from $v_a$ to $v_1$, \ldots, $v_{a-1}$ in order.

Unfortunately, this is incompatible with the desire to process all edges incident with a few of the vertices,
as needed for the serialization described in the previous section.  As an extreme example,
for the complete graph $K_n$, if we used the order of edges described in the previous paragraph,
we would first finish processing all edges incident with a vertex when considering
the edge $v_nv_1$, after $\binom{n-1}{2}$ other edges have been processed.

We restrict ourselves to the type of edge orderings described in the previous section,
i.e., we fix the ordering $v_1$, \ldots, $v_n$ of the vertices and process all edges incident
with $v_1$, then all the remaining edges incident with $v_2$, etc.  In choosing the vertex
ordering, we now aim to balance two somewhat contradictory goals: We would like
the subgraph induced by each initial segment $v_1$, \ldots, $v_a$ to be as hard to orient
as possible, but also keep the number of not yet processed vertices with neighbors in this segment small, to limit
the number of choices for the degrees of the variables corresponding to $v_{a+1}$, \ldots, $v_n$.

The latter objective suggests to proceed according to an ordering with the smallest vertex separation number\footnote{The
\emph{vertex separation number} of an ordering of vertices of a graph is the minimum integer $k$ such that for every vertex $v$, at most $k$ vertices
appearing after $v$ in the ordering are adjacent to $v$ or vertices preceding $v$ in the ordering.
The vertex separation number of a graph is the minimum of vertex separation numbers of the orderings of its vertices.}.
However, the vertex separation number of a graph is equal to its pathwidth, and thus it is NP-hard to
determine exactly, or approximate up to a constant additive term~\cite{BGHKTree}; and while pathwidth
is fixed-parameter tractable~\cite{pathwidthfpt}, the corresponding algorithm is not useful in practice.
Moreover, following this ordering may conflict with the desire to keep the processed part of the graph
as dense (and thus hard to orient) as possible.

\begin{table}
\begin{center}
\begin{tabular}{|l|c|c|c|c|c|c|c|}
\hline
Data set&VSEP&MD&MD+PROC&OVER&LIST&LIST+DEG&MDR\\
\hline
crit12&59\%&61\%&59\%&64\%&94\%&85\%&88\%\\
planar&55\%&61\%&51\%&49\%&45\%&45\%&50\%\\
\hline
\end{tabular}
\end{center}
\caption{Comparison of ordering heuristics.}\label{tab-ordering}
\end{table}

Hence, we resorted to choosing an ordering heuristically; we compare several of heuristics that we tested in Table~\ref{tab-ordering}.
The first six heuristics greedily select $v_1$, $v_2$, \ldots
in order, always choosing the vertex $v_i$ according to one of the following rules.
\begin{itemize}
\item VSEP: So that after processing $v_i$, the number of vertices with at least one processed
neighbor is minimized.  Among those vertices for which this number is smallest, the one of minimum degree is chosen.
\item MD: Minimum degree in $G-\{v_1,\dots,v_{i-1}\}$.
\item MD+PROC: As MD, but secondarily among the vertices of the minimum degree, the
one with most neighbors in $\{v_1,\ldots,v_{i-1}\}$.
\item OVER: Minimizing the number of edges leaving the initial segment, i.e., with minimum
difference of the degree in $G-\{v_1,\dots,v_{i-1}\}$ and the number of neighbors in $\{v_1,\ldots,v_{i-1}\}$.
\item LIST: Smallest list size, and secondarily smallest degree in $G-\{v_1,\dots,v_{i-1}\}$.
\item LIST+DEG: Smallest sum of the list size and the degree in $G-\{v_1,\dots,v_{i-1}\}$.
\end{itemize}
Finally, MDR takes the reverse of the ordering obtained by MD.

For the data sets, crit12 is a set of 844030 graphs with list sizes with at most 23 edges each,
with no non-zero coefficients relevant for Theorem~\ref{thm-alontarsi} (obtained as part of a project to generate
obstructions to 5-choosability of graphs drawn on the torus); planar is a choosable planar graph
with 16 vertices and 37 edges (planar $3\times 4$ grid with diagonals and with four additional vertices adjacent to the first row,
last row, first column, and last column, the vertices incident with the outer face with list size three
and all others with list size five).  For the second graph, we disable the serialization to eliminate the
(semi-random) effect of stopping early when the first non-zero coefficient is found in one of the branches.
In the table, we list the number of monomials with non-zero coefficients obtained during the whole computation,
relatively to the results obtained when we simply retain the ordering of the vertices as in the input.

On several other test graphs the choice of heuristic did not play any role (note for example
that all the heuristics except for VSEP and MDR result in the same ordering on regular graphs where
all vertices have the same list size).  Based on these experiments, we chose to use MD+PROC heuristic
to determine the order of vertices to process---in addition to behaving quite well on both data sets
described above, it has the advantage of improving the granularity of the sequentialization.
That is, note that in the algorithm, we can only check whether the bound $N$ on the number of coefficients
is exceeded after we have processed all edges leaving a vertex.  If the input graph is $d$-degenerate,
then MD+PROC (as well as MD) heuristics guarantee that at most $d$ edges are processed between
the consecutive checks.

Let us remark that there exist graphs for which the VSEP heuristics significantly outperforms MD+PROC;
see the end of the following subsection for an example.  Hence, in specific applications, it may be worth
experimenting with the choice of the ordering.

There seems to be a lot of room for improvement in the choice of the ordering.  For example, one might be able
to come up with a sufficiently efficient way to compute or estimate for a given subgraph $H$ the number of different
vectors $f\prec s$ such that $H$ has an $f$-orientation.  It would then be natural to choose the $i$-th vertex $v_i$
of the ordering so that this quantity is minimized among the subgraphs consisting of edges incident with $v_1$, \ldots, $v_i$.
It might also be useful not to just select the vertices one by one, but take the effects of processing several vertices
into account when choosing $v_i$.  Finally, let us note that we could select the ordering adaptively, possibly choosing
a different ordering in each of the branches of the serialization.

\subsection{Testing results}\label{sec-testing}

We have tested the performance of the described algorithm on a number of test cases.  The measurements
were performed on a machine with Intel Core i5-7200U 2.50GHz CPU with 8GB of memory.

\begin{itemize}
\item The crit12 set of 844030 graphs described in the previous section, with no non-zero coefficients relevant for Theorem~\ref{thm-alontarsi},
was processed in 4 seconds.
\item The non-choosable graphs obtained from $a$ copies of $K_b$ glued in a path-like fashion over distinct vertices, with the list
sizes equal to the vertex degrees.  We performed the tests for $b\in \{3,4,5,6\}$ and various values of $a$; the dependence
of the runtime on the number of edges (rounded to the nearest multiple of 5) for given $b$ is shown in Table~\ref{tab-akb}.
This gives an indication of how the performance of the method degrades for denser graphs.

\begin{table}
\begin{center}
\begin{tabular}{|c|cccccc|}
\hline
\diaghead(-2,1){aaaaaaaa}{$b$}{edges}&45&50&55&60&65&70\\
\hline
3&$<$0.01 s&0.01 s&0.08 s&0.12 s&0.3 s&1 s\\
4&&0.15 s&0.7 s&4 s&26 s&186 s\\
5&&1 s&&23 s&&703 s\\
6&0.5 s&&&50 s&&\\
\hline
\end{tabular}
\end{center}
\caption{Runtime of the standard Alon-Tarsi on graphs $aK_b$.}\label{tab-akb}
\end{table}

\item We also tested several classes of choosable graphs.  Due to early termination when a suitable coefficient is
found, the timing in this case can be expected to be more affected by the luck (whether the choice of the ordering
of the vertices and corresponding branching in the sequentialization quickly leads to a non-zero coefficient).
\begin{itemize}
\item The graph obtained from $a$ copies of $K_b$ ($b\ge 4$) glued in a path-like fashion over distinct vertices and with one edge removed
from one of the copies, with the list sizes equal to the vertex degrees.  Here we list some of the timing results that we find indicative
of the performance (in all the cases, $a$ is chosen largest such that the graph has at most $2^7$ vertices):

\begin{center}
\begin{tabular}{|cccc|}
\hline
$b$&$a$&edges&time\\
\hline
4&42&251&0.2 s\\
5&31&309&0.4 s\\
6&25&374&0.5 s\\
\hline
\end{tabular}
\end{center}

\item Planar $a\times a$ grid with diagonals and with four additional vertices adjacent to the first row,
last row, first column, and last column, the vertices incident with the outer face with list size three
and all others with list size five.  It is known that the choosability of these graphs can be shown using Theorem~\ref{thm-alontarsi},
as proved by Zhu~\cite{atar}.  For $a=11$ (364 edges), the program runs for about 0.6 s.
\item Cycle with vertices $0$, $1$, \ldots, $3n-1$ together with all edges of form $\{i,(i+n)\bmod 3n\}$, and with all
list sizes equal to 3.  Note that this graph is an edge-disjoint union of a cycle of length $3n$ and of $n$ triangles;
and thus it is known to be $3$-choosable using Theorem~\ref{thm-alontarsi} by the well-known
cycle plus triangles theorem~\cite{cytri}.  In this case, already for $n=15$ (90 edges), the program takes around 9 seconds.
More importantly, at this point the program requires around 7GB of memory.
\end{itemize}
\end{itemize}

The worse performance in the last case is worth some discussion, in particular as the memory consumption indicates
a failure of the serialization improvement.  Let us note that there is only a unique monomial useful for Theorem~\ref{thm-alontarsi},
$x^f$ for $f$ assigning the value $2$ to all vertices.  In particular, when performing sequentialization, only one of the branches
contains relevant monomials.  Somewhat luckily, the ordering of the monomials chosen in the program
is such that when we perform serialization, we first consider the monomials whose restrictions to the already processed
vertices have lexicographically largest degrees, i.e., those where all variables corresponding to the processed vertices have degree two.
Consequently, we follow this unique relevant branch first and avoid any backtracking.
However, the issue is that even with this restriction, the number of monomials with non-zero coefficients is too large---the number of
monomials that the program tracks peaks after 15 vertices are processed at more than $47$ million\footnote{We tested the possibility to remove the monomials that cannot contribute to the coefficient of $x^f$, detected by the reduction to
the maximum matching in bipartite graphs described at the beginning of Section~\ref{sec-direct}.  This reduces the peak number
of monomials to less than $26$ million, roughly halving the memory consumption.  However, the extra overhead associated
with this pruning increases the running time to 36 seconds, i.e., by factor of four.  Thus, even in what arguably could be seen
as very favorable circumstances for the pruning, it does not seem to be worthwhile.  Let us however remark that
using the pruning eliminates the element of chance of the right branch happening to be taken first in the serialization.}.

A reason for this fast growth is that the chosen ordering of vertices ends up having very large vertex separation number.
Consequently, for this particular graph, the VSEP ordering heuristic siginificanly outperforms others, as it leads
us to process the graph in the ``triangle by triangle`` ordering with vertex separation number six.  In the VSEP
ordering, the running time is less than 0.01 s even for $n=40$ (240 edges).  In the case that the triangles are added to the
cycle at random, VSEP still outperforms MD+PROC, but to a lesser degree: MD+PROC starts to run out of memory at $n=26$
(at which point VSEP ordering leads to around 9 times fewer monomials to process), while VSEP starts to run out of memory at $n=34$.

\section{Extended Alon-Tarsi method}\label{sec-gener}

Let us start by proving the theoretical result underlying our strengthening of Alon-Tarsi method.

\begin{proof}[Proof of Theorem~\ref{thm-main}]
Since both $L$-colorability and the constraint from the statement of the theorem only depend
on the equality between colors, without loss of generality, we can assume that the elements of $\bigcup_{v\in V(G)} L_v$
are algebraically independent.  For each $v\in V(G)$, let $\ell_v=|L_v|$ and
\begin{equation}\label{eq-first}
p_v=x_v^{\ell_v} - \prod_{c\in L_v} (x_v-c)=\Bigl(\sum L_v\Bigr)\cdot x^{\ell_v-1}+\ldots.
\end{equation}
Note that $p_v(c)=c^{\ell_v}$ for every $c\in L_v$.

Let $q$ be the polynomial obtained as follows: Start with $q=P_G$ and while there exists a vertex $v\in V(G)$
such that $\deg_{x_v} q\ge \ell_v$, replace $q$ by
$$\sum_{f\in \mathbb{Z}^{V(G)}, f(v)<\ell_v} [x^f]\,q\cdot x^f + \sum_{f\in \mathbb{Z}^{V(G)}, f(v)\ge \ell_v} [x^f]\,q \cdot x^{f-\ell_v1_v} \cdot p_v.$$
Note that each such replacement decreases $\deg_{x_v} q$ by at least one, and moreover, that
$q(\varphi)=P_G(\varphi)$ for every $\varphi\in \bigtimes_{v\in V(G)} L_v$.
We end up with a polynomial $q$ such that $\deg_{x_v} q<\ell_v$ for each $v\in V(G)$.

Note that $P_G$ is homogeneous and a product of $|E(G)|$ terms, an thus
$\sum_{v\in V(G)} f'(v)=|E(G)|$ for every $f'\in\mathbb{Z}^{V(G)}$ such that $[x^{f'}]\,P_G\neq 0$.
Since $\sum_{v\in V(G)} f(v)=|E(G)|-1$, the coefficient at $x^f$ in $q$ comes from the replacement
of $x^{\ell_v}$ by $p_v$ in the monomials $x^{f+1_v}$ of $P_G$ for vertices $v\in V(G)$ such that $f(v)=\ell_v-1$.
By (\ref{eq-first}), each such monomial contributes $[x^{f+1_v}]\,P_G\cdot \sum L_v$ to $[x^f]\,q$.
Moreover, by Theorem~\ref{thm-alontarsi}, since $G$ is not $L$-colorable we have $[x^{f+1_v}]\,P_G=0$
for every $v\in V(G)$ such that $f(v)<\ell_v-1$.  Therefore,
$$[x^f]\,q=\sum_{v\in V(G)} [x^{f+1_v}]\,P_G \cdot \sum L_v.$$
We have $q=0$, as otherwise Theorem~\ref{thm-nullstellensatz} would imply that there exists $\varphi\in \bigtimes_{v\in V(G)} L_v$
such that $q(\varphi)=P_G(\varphi)\neq 0$, and thus by Observation~\ref{obs-color}, $G$ would be $L$-colorable.
Therefore, letting $C=\bigcup_{z\in V(G)} L_z$, we have
\begin{align*}
0&=[x^f]\,q=\sum_{v\in V(G)} [x^{f+1_v}]\,P_G \cdot \sum L_v
=\sum_{v\in V(G)} [x^{f+1_v}]\,P_G \cdot \sum_{c\in C} c\cdot \chi_{L,c}(v)\\
&=\sum_{c\in C} c \cdot \sum_{v\in V(G)} [x^{f+1_v}]\,P_G\cdot \chi_{L,c}(v).
\end{align*}
Since the colors are algebraically independent, the coefficient
$$\sum_{v\in V(G)} [x^{f+1_v}]\,P_G\cdot \chi_{L,c}(v)$$
has to be zero for each color $c$.
\end{proof}

For a function $s:V(G)\to\mathbb{N}$, let us say that a function $f:V(G)\to\mathbb{N}$ is \emph{$s$-tight} if $f(v)\le s(v)$ for every $v\in V(G)$
and $f(v)=s(v)$ for exactly one vertex $v\in V(G)$.  For an $s$-tight function $f$, the \emph{$s$-base} of $f$ is the function $f':V(G)\to\mathbb{N}$
such that $f'(v)=f(v)$ for every vertex $v\in V(G)$ such that $f(v)<s(v)$ and $f'(v)=f(v)-1$ for the unique vertex $v\in V(G)$ such that $f(v)=s(v)$.
Note that for each monomial $x^f$ whose coefficient appears on the left-hand side of the equality from Theorem~\ref{thm-main},
either $f\prec s$ or $f$ is $s$-tight.  If $[x^f]\,P_G\neq 0$ for some $f\prec s$, then the graph is $s$-choosable by Theorem~\ref{thm-alontarsi}.
Hence, to apply Theorem~\ref{thm-main}, we need to focus on the coefficients of the monomials $x^f$ where $f$ is $s$-tight.
We then group them according to their $s$-bases and each group gives us a linear constraint for the characteristic vectors.

To enumerate the coefficients, we use the algorithm described in Section~\ref{sec-impl}, modified to also include monomials with the $s$-tight degree vectors
in the truncation.  Let us give a few remarks on the implementation:
\begin{itemize}
\item It is convenient to store the polynomials as lists of triples $(f',v,c)$, where $v$ is either a vertex of $G$ or $\varnothing$,
$f'\prec s$ and $c$ is the coefficient at
\begin{itemize}
\item $x^{f'}$ if $v=\varnothing$, and at
\item $x^{f'+1_v}$ if $v\neq \varnothing$; in this case $f'(v)=s(v)-1$.
\end{itemize}
The list is sorted in the increasing lexicographic ordering primarily according to $f$ and secondarily according to $v$.
This way, the monomials with the same $s$-base appear together in the list representing the result and we can form constraints from them directly,
without further post-processing.
\item In sequentialization, we partition the monomials only according to $f'$, not taking $v$ into account.  This ensures that all monomials
with the same $s$-base end up in the same branch, and thus we can still process the branches independently.
\item However, this representation introduces an issue with preserving the ordering on the output: In the standard Alon-Tarsi,
we relied on the fact that if $f_1<f_2$, then for any vertex $v\in V(G)$, $f_1+1_v<f_2+1_v$, see Section~\ref{ssec-repr} for details  However, the analogue that we would
need with the representation described above does not hold.  Indeed, consider elements $(f_1,\varnothing,c_1)$ and $(f_2,\varnothing,c_2)$
such that $f_1<f_2$, $f_1(v)=s(v)-2$ and $f_2(v)=s(v)-1$.  When processing an edge incident with $v$, these elements
will contribute $(f_1+1_v,\varnothing,c_1)$ and $(f_2,v,c_2)$, and it is not necessarily the case that $f_1+1_v\le f_2$.
To avoid this issue, when processing an edge $uv$, we produce the output by merging the following four streams,
each of which is guaranteed to be increasing:
\begin{itemize}
\item $(f+1_v,w,c)$ for $(f,w,c)$ in the input such that $f(v)\le s(v)-2$,
\item $(f,v,c)$ for $(f,\varnothing,c)$ in the input such that $f(v)=s(v)-1$,
\item $(f+1_u,w,-c)$ for $(f,w,c)$ in the input such that $f(u)\le s(u)-2$, and
\item $(f,u,-c)$ for $(f,\varnothing,c)$ in the input such that $f(u)=s(u)-1$.
\end{itemize}
\item As the constraints given by Theorem~\ref{thm-main} are linear, they can be represented compactly: We can
ignore those that are linearly dependent on the others, and thus it suffices to store a list of at most $|V(G)|$ linearly
independent ones.  Moreover, this list can of course be updated incrementally as new constraints are generated;
hence, we do not need to generate the full list of constraints first (this is particularly important for the serialization,
as being forced to store all the coefficients of the result would significantly limit its usefulness in decreasing the memory consumption).

As a minor remark, in our implementation we test the linear dependence over a finite field $\mathbb{F}_p$ for a prime $p=2^{32}-1$,
to avoid the issues associated with the growth of the coefficients during Gaussian elimination over $\mathbb{Q}$.  This may (extremely rarely)
lead to some of the constraints being dropped unnecessarily, as they are linearly dependent over $\mathbb{F}_p$ but not over $\mathbb{Q}$.
\end{itemize}

Once the constraints have been collected, we need to find the characteristic vectors that satisfy them, i.e., to
list the $\{0,1\}$-solutions to the system of linear equations.  This can be achieved using any integer linear programming solver.
As the number of variables is typically rather small, even the simple approach of performing the Gaussian elimination,
then going over all $\{0,1\}$-choices for the free variables and checking whether the values of the variables
determined by the system also belong to $\{0,1\}$ is often viable.

If the number of possible characteristic vectors is reasonably small, we can then form another integer linear program
expressing that their linear combination (with non-negative integral coefficients) should be equal to the vector
of the prescribed list sizes.  Each solution to this program then corresponds to a list assignment from which the
graph is not necessarily colorable.  Our implementation again only uses the simple approach of performing the Gaussian elimination and
checking all possible choices for the free variables, but a more sophisticated integer linear programming solver would be more appropriate
for this part.

As the final step, if the number of possible list assignments is not too large, one can try coloring the graph from
each of them (as the number of vertices is typically rather small, any CSP or SAT solver, or even rather simple
exhaustive coloring algorithms, should be good enough for the task).  In our implementation, we express the
colorability from the lists as a satisfiability problem in CNF and use MiniSat~\cite{een2005minisat}
SAT solver for this part.

\subsection{Testing results}

Extended Alon-Tarsi method of course requires more of the coefficients of the graph polynomial to be computed
compared to the standard Alon-Tarsi.  Also, we cannot terminate as soon as a monomial with a non-zero coefficient
is found, unless this monomial shows that the graph is $s$-choosable by Theorem~\ref{thm-main}, in which case
the extended Alon-Tarsi method was not actually needed\footnote{While we actually can terminate the enumeration at
any point, this could mean that we miss some of the constraints that would be found later (though this still might
be a valid alternative in case the search takes too long, and it might be reasonable to stop the search when
the space of solutions did not change for a long time).  The time measurements we report are for the
full run of the extended Alon-Tarsi that enumerates all of the tight monomials.}.  Thus, we should expect
the time complexity (and possibly the memory consumption, in cases where the sequentialization does not help) to be
worse compared to the standard Alon-Tarsi method.

To illustrate this, consider the non-choosable graphs obtained from $a$ copies of $K_b$ glued in a path-like fashion over distinct vertices, with the list
sizes equal to the vertex degrees; the evaluation of the standard Alon-Tarsi method for these graphs can be found
at the beginning of Section~\ref{sec-testing}.  We provide the dependence of the time on the number of edges
(rounded to the nearest multiple of 5) for given $b$ in Table~\ref{tab-akbe}.  The time includes collecting the constraints and eliminating the linearly
dependent ones, but no further processing (listing the feasible characteristic vectors, \ldots).
Comparing this with the results for the standard Alon-Tarsi shown in Table~\ref{tab-akb}, we conclude that for this particular type of graphs,
extended Alon-Tarsi needs about as much time as the standard Alon-Tarsi on a graph with about 15 more edges.

\begin{table}
\begin{center}
\begin{tabular}{|c|ccccc|}
\hline
\diaghead(-2,1){aaaaaaaa}{$b$}{edges}&30&35&40&45&50\\
\hline
3&0.03 s&0.2 s&0.5 s&1 s&4 s\\
4&0.1 s&0.5 s&3.5 s&&29 s\\
5&0.1 s&&3 s&&128 s\\
6&0.1 s&&&19 s&\\
\hline
\end{tabular}
\end{center}
\caption{Runtime of the extended Alon-Tarsi on graphs $aK_b$.}\label{tab-akbe}
\end{table}

For a more real-world test case, we evaluated the performance on the data set crit12 of 844030 graphs with prescribed list sizes $s$
(see also Section~\ref{ssec-order}).  The graphs in this data set were obtained as part of a project to generate
obstructions to 5-choosability of graphs drawn on the torus, and were selected through various heuristics (including checking
that they have no non-zero coefficients usable in Theorem~\ref{thm-alontarsi}) as candidates
for being critical, i.e., with the property that they are not $s$-choosable, but all their proper subgraphs are $s$-choosable.
Hence, they are generally quite close to being $s$-choosable, and thus one should expect Theorem~\ref{thm-main}
to give interesting information about their colorability.  For each of these graphs, we
\begin{itemize}
\item run standard Alon-Tarsi method (and check that it does not apply),
\item gather the constraints of the extended Alon-Tarsi method,
\item list the feasible characteristic vectors and use them to find ``deletable'' edges $uv$
with the property that in every $s$-list-assignment $L$ such that the graph is not $L$-colorable, the lists of $u$ and $v$ are disjoint,
\item if any edges are deletable, re-run the standard Alon-Tarsi method on the graph obtained
by deleting them to check its $s$-choosability,
\item list up to 100 possible $s$-list-assignments where the characteristic vectors are feasible,
\item if less than 100 such $s$-list-assignments exist, use the MiniSat solver to determine whether
the graph is colorable from any of them.
\end{itemize}

The total run time was 62 s, with 5\% taken by standard Alon-Tarsi (including a few re-runs), 21\% by extended Alon-Tarsi,
4\% by finding the feasible vectors and deletable edges, 29\% by listing the $s$-list-assignments and 35\% by coloring using
the MiniSat solver (with the remaining 6\% taken by reading the input and other miscellaneous overheads).
The results were as follows:
\begin{itemize}
\item For 548502 (65\%) of the graphs, the program found an $s$-list-assignment from which they cannot be colored.
\item For 19322 (2\%) of the graphs, the program found deletable edges and decided that they are $s$-choosable by
re-running the standard Alon-Tarsi after deleting them.
\item For 732 of the graphs, the program proved that they are $s$-choosable directly (78 have no feasible vectors,
for 215 of them the vectors do not combine to any $s$-list-assignment, and 439 were shown to be colorable from all
possible $s$-list-assignment by the MiniSat solver).
\item For 11554 (1.4\%) of the graphs, the program only succeeded in finding deletable edges (thus showing that the
graphs are not critical).
\item For 263526 (31\%) of the graphs, the program reached no conclusion, as they have too many possible $s$-list-assignments to explore.
\item For 394 of the graphs, the program reached no conclusion as all the coefficients relevant for Theorem~\ref{thm-main} are zero.
\end{itemize}

It is noteworthy that the extended Alon-Tarsi method failed to provide any information at all only for a negligible fraction
(less than 0.1\%) of the graphs from this data set, indicating the usefulness of the method in similar circumstances (enumeration of critical
graphs, confirming reducibility of configurations in graphs drawn on a fixed surface).

For this data set,
the steps of listing the potential bad assignments and checking the colorability from them take disproportionate
amount of time relative to their usefulness, and so perhaps running them is not worthwhile, unless a confirmation
that the graph is not $s$-choosable is needed.  On the other hand, it should be noted that we primarily focused
on optimizing the implementations of the standard and extended Alon-Tarsi method, and thus the running time of the other
parts definitely can be improved.

\section{Limitations and extensions}\label{sec-more}

Let us note one important restriction for the extended Alon-Tarsi method:  Let $s$ be the function assigning list sizes to vertices of a graph $G$.
Suppose that a monomial $x^f$, where $f$ is $s$-tight, has a non-zero coefficient. Let $v$ be the vertex such that $f(v)=s(v)$.
Letting $s'=s+1_v$, Theorem~\ref{thm-alontarsi} implies that $G$ is $s'$-choosable.
Hence, Theorem~\ref{thm-main} can only give interesting information for graphs $G$
that are``close to $s$-choosable'' in the sense that $G$ is $(s+1_v)$-choosable for at least one vertex $v$.

Moreover, while the examples and the experiments discussed in the previous section show that Theorem~\ref{thm-main}
applies for many interesting combinations of graphs and list sizes, there are
rather simple (and close to $s$-choosable) graphs for which Theorem~\ref{thm-main} does not give any information, since all
relevant coefficients turn out to be zero; see e.g. Example~\ref{ex-wh} below.

Note also that the constraints on the characteristic vectors of the colors obtained using Theorem~\ref{thm-main}
are necessarily linear, and consequently if feasible characteristic vectors satisfy more complicated conditions,
the method can at most return a linear relaxation of these conditions.

It is natural to mitigate these concerns by using further coefficients of the graph polynomial.
Indeed, the argument used to prove Theorem~\ref{thm-main} clearly can be pushed further; for example, as the next
step, one obtains the following result.
\begin{theorem}\label{thm-next}
Let $G$ be a graph, let $L=\{L_v\subset \mathbb{R}:v\in V(G)\}$ be an assignment of lists to vertices of $G$,
and suppose that $f:V(G)\to \mathbb{Z}$ satisfies $f(v)<|L_v|$ for every $v\in V(G)$ and $\sum_{v\in V(G)} f(v)=|E(G)|-2$.
If $G$ is not $L$-colorable, then
\begin{equation}\label{eq-singleq}
\sum_{\{u,v\}\subseteq V(G)} [x^{f+1_u+1_v}]\,P_G \cdot \chi_{L,c}(u)\chi_{L,c}(v) + \sum_{v\in V(G)} [x^{f+2\cdot 1_v}]\,P_G \cdot \chi_{L,c}(v)=0
\end{equation}
for each color $c$ and 
\begin{align}\label{eq-doubleq}
&\sum_{\{u,v\}\subseteq V(G)} [x^{f+1_u+1_v}]\,P_G \cdot (\chi_{L,c_1}(u)\chi_{L,c_2}(v)+\chi_{L,c_2}(u)\chi_{L,c_1}(v))\nonumber\\
&+\sum_{v\in V(G)} [x^{f+2\cdot 1_v}]\,P_G \cdot \chi_{L,c_1}(v)\chi_{L,c_2}(v)=0
\end{align}
for any distinct colors $c_1$ and $c_2$.
\end{theorem}
\begin{proof}
Let $\ell_v$ and $p_v$ for $v\in V(G)$ and $q$ be as in the proof of Theorem~\ref{thm-main}.
Let $S=\{v\in V(G):f(v)=\ell_v-1\}$ and $D=\{v\in V(G):f(v)=\ell_v-2\}$.  Observe that
\begin{align*}
[x^f]\,q&=\sum_{\{u,v\}\subseteq S} [x^{f+1_u+1_v}]\,P_G \cdot \sum L_u\sum L_v
+\sum_{v\in S} [x^{f+2\cdot 1_v}]\,P_G \cdot \Bigl(\sum L_v\Bigr)^2\\
&\phantom{=}-\sum_{v\in S\cup D} [x^{f+2\cdot 1_v}]\,P_G \cdot \sum_{\{c_1,c_2\}\subseteq L_v} c_1c_2.
\end{align*}
Since $G$ is not $L$-colorable, we have $[x^{f+1_u+1_v}]\,P_G=0$ for distinct $u, v\not\in S$ by Theorem~\ref{thm-alontarsi}
and $\sum_{v\in V(G)} [x^{f+1_u+1_v}]\,P_G \cdot \sum L_v=0$ for $u\not\in S$ by Theorem~\ref{thm-main}.
Therefore,
\begin{align*}
\sum_{\{u,v\}\subseteq V(G)}& [x^{f+1_u+1_v}]\,P_G \cdot \sum L_u\sum L_v\\
&=\sum_{\{u,v\}\subseteq S} [x^{f+1_u+1_v}]\,P_G \cdot \sum L_u\sum L_v
+\sum_{u\in V(G)\setminus S}\sum L_u\cdot \sum_{v\in V(G)}[x^{f+1_u+1_v}]\,P_G \cdot \sum L_v\\
&\phantom{=}-\sum_{\{u,v\}\subseteq V(G)\setminus S}[x^{f+1_u+1_v}]\,P_G \cdot \sum L_u\sum L_v
-\sum_{v\in V(G)\setminus S}[x^{f+2\cdot 1_v}]\,P_G \cdot \Bigl(\sum L_v\Bigr)^2\\
&=\sum_{\{u,v\}\subseteq S} [x^{f+1_u+1_v}]\,P_G \cdot \sum L_u\sum L_v
-\sum_{v\in V(G)\setminus S}[x^{f+2\cdot 1_v}]\,P_G \cdot \Bigl(\sum L_v\Bigr)^2.
\end{align*}
Moreover, $[x^{f+2\cdot 1_v}]\,P_G=0$ for $v\not\in S\cup D$.  Hence, we obtain
\begin{align*}
0=[x^f]\,q&=\sum_{\{u,v\}\subseteq V(G)} [x^{f+1_u+1_v}]\,P_G \cdot \sum L_u\sum L_v\\
&\phantom{=}+\sum_{v\in V(G)} [x^{f+2\cdot 1_v}]\,P_G \cdot \Bigl(\Bigl(\sum L_v\Bigr)^2-\sum_{\{c_1,c_2\}\subseteq L_v} c_1c_2\Bigr).
\end{align*}
Since the colors are algebraically independent, for each color $c$ the coefficient
$$\sum_{\{u,v\}\subseteq V(G)} [x^{f+1_u+1_v}]\,P_G \cdot \chi_{L,c}(u)\chi_{L,c}(v) + \sum_{v\in V(G)} [x^{f+2\cdot 1_v}]\,P_G \cdot \chi_{L,c}(v)$$
at $c^2$ must be zero, and for any distinct colors $c_1$ and $c_2$, the coefficient
\begin{align*}
&\sum_{\{u,v\}\subseteq V(G)} [x^{f+1_u+1_v}]\,P_G \cdot (\chi_{L,c_1}(u)\chi_{L,c_2}(v)+\chi_{L,c_2}(u)\chi_{L,c_1}(v))\\
&+\sum_{v\in V(G)} [x^{f+2\cdot 1_v}]\,P_G \cdot \chi_{L,c_1}(v)\chi_{L,c_2}(v)
\end{align*}
at $c_1c_2$ must be zero.
\end{proof}

Let us give a simple example of a graph where Theorem~\ref{thm-next} gives more information than Theorem~\ref{thm-main}.

\begin{example}\label{ex-wh}
Let $G$ consist of the wheel with center $0$ and spokes $1$, \ldots, $4$, with a list assignment $L$ such that
$|L_0|=3$ and $|L_1|=\ldots=|L_4|=2$.  A straightforward case analysis shows that up to permutation of colors,
there are only a few list assignments of this size such that $G$ is not $L$-colorable: $L_0=\{a,b,c\}$ and,
letting $L_i=L_{i-4}$ for $i\ge 5$,
\begin{itemize}
\item for some $i\in\{1,2\}$, $L_i=L_{i+1}=\{a,b\}$ and $L_{i+2}=L_{i+3}=\{a,c\}$, or
\item for some $i\in\{1,2,3,4\}$, $L_i=L_{i+1}=\{a,b\}$ and $L_{i+2}=L_{i+3}=\{c,d\}$, or
\item for some $i\in\{1,2,3,4\}$, $L_i=\{a,d\}$, $L_{i+1}=\{a,b\}$, $L_{i+2}=\{b,c\}$ and $L_{i+3}=\{c,d\}$.
\end{itemize}
However, the graph polynomial of $G$ does not have any monomials with non-zero coefficient relevant for Theorem~\ref{thm-main}
(indeed, to apply Theorem~\ref{thm-main} we would have to have $f(0)<3$, $f(1),\ldots, f(4)<2$ and
$f(0)+\ldots+f(4)=|E(G)|-1=7$, which is not possible).  Let us remark that
\begin{align*}
P_G=x_0^2x_1x_2x_3x_4(&x_1^2+x_2^2+x_3^2+x_4^2\\
&-2x_1x_2-2x_2x_3-2x_3x_4-2x_4x_1\\
&+x_1x_3+x_2x_4\\
&+x_0x_1+x_0x_2+x_0x_3+x_0x_4\\
&-2x_0^2) + \ldots,
\end{align*}
and Theorem~\ref{thm-next} thus implies that
the characteristic vector $\chi$ of any color must satisfy
\begin{align*}
\chi(1)+\chi(2)+\chi(3)+\chi(4)&\\
-2\chi(1)\chi(2) -2\chi(2)\chi(3) -2\chi(3)\chi(4) -2\chi(4)\chi(1)&\\
+\chi(1)\chi(3)+\chi(2)\chi(4)&\\
+\chi(0)\chi(1) +\chi(0)\chi(2) +\chi(0)\chi(3) +\chi(0)\chi(4)&\\
-2\chi(0)&=0,
\end{align*}
and that the characteristic vectors $\chi$ and $\chi'$ of any two distinct colors must satisfy
\begin{align*}
\chi(1)\chi'(1)+\chi(2)\chi'(2)+\chi(3)\chi'(3)+\chi(4)\chi'(4)&\\
-2\sum_{(i,j)\in \{(1,2),(2,3),(3,4),(4,1)\}} \chi(i)\chi'(j)+\chi'(i)\chi(j)&\\
+\chi(1)\chi'(3)+\chi'(1)\chi(3) +\chi(2)\chi'(4)+\chi'(2)\chi(4)&\\
+\sum_{j\in\{1,2,3,4\}} \chi(0)\chi'(j)+\chi'(0)\chi(j)&\\
-2\chi(0)\chi'(0)&=0.
\end{align*}
This suffices to restrict the possible list assignments of lists of size $3,2,2,2,2$ in order
to one of 22 options, ten of them being the assignments $L$ from which $G$ is not $L$-colorable described above.
\end{example}

Theorem~\ref{thm-next} puts more complicated quadratic constraints on the feasible characteristic
vectors of colors, and moreover, it put constraints on the feasible characteristic vectors of
pairs of colors that can appear in the list assignment, thus possibly revealing much more information
than the single-color linear constraints from Theorem~\ref{thm-main}.  On the flip side,
the computational cost of applying Theorem~\ref{thm-next} would be substantially higher---more
coefficients of the graph polynomial have to be computed, and it is not clear how to represent
the obtained constraints compactly or how to obtain information from them about the feasible characteristic vectors.

Perhaps the most practical option would be to first enumerate the vectors that are feasible
according to Theorem~\ref{thm-main} as discussed in Section~\ref{sec-gener} and build a complete
graph $K$ (with loops) whose vertices are the feasible characteristic vectors; one would expect $K$ to be
reasonably small for graphs with about 10--20 vertices.  As constraints from Theorem~\ref{thm-next} are generated,
we delete the vertices of $K$ violating (\ref{eq-singleq}) and edges of $K$ violating (\ref{eq-doubleq}).
In the end, for any list assignment $L$ with given list sizes such that $G$ is not $L$-colorable,
the characteristic vectors of the colors in $L$ must form a clique in $K$, and if two distinct colors
have the same characteristic vector, it must form a loop in $K$.

And of course, we can push the argument further, obtaining cubic, quartic, \ldots constraints.
However, it is not clear whether implementing this strengthening would be worth the effort, especially
given that the method from Section~\ref{sec-gener} seems to perform quite well in practice.

\bibliography{../data.bib}

\end{document}